\documentclass{article}

\usepackage{amsmath}
\usepackage{enumitem}
\usepackage{amsthm}

\newtheorem{theorem}{Theorem}
\newtheorem{lemma}{Lemma}
\newtheorem{corollary}{Corollary}
\newtheorem{definition}{Definition}

\usepackage{amsfonts}
\usepackage{amssymb}
\usepackage{graphicx}
\usepackage{appendix}

\usepackage{caption}
\usepackage{makecell}
\usepackage{hhline}
\usepackage[ruled,vlined,linesnumbered,boxed]{algorithm2e}

\usepackage{xcolor}

\title{New results for the detection of bicliques} 

\author{
    George Manoussakis* \\
    \small *LI-PARAD, Université de Versailles Saint-Quentin-En-Yvelines, Paris Saclay
}
\date{}

\begin{document}

\maketitle

\begin{abstract}

Building on existing algorithms and results, we offer new insights and algorithms for various problems related to detecting maximal and maximum bicliques. Most of these results focus on graphs with small maximum degree, providing improved complexities when this parameter is constant; a common characteristic in real-world graphs.

\end{abstract}

\section{Introduction}

Identifying all bicliques in a graph has been a notable area of research because of its varied applications. This problem is significant in theoretical computer science and also contributes to fields such as data mining, bioinformatics, and social network analysis, where comprehending the complex relationships between entities is essential.

For instance, in computational biology, enumerating bicliques has various applications, particularly in biclustering—a data mining technique that simultaneously clusters rows and columns of a matrix to uncover submatrices of related data. This process is closely linked to biclique enumeration in graph theory, where each biclique corresponds to a bicluster, representing strongly connected subsets of entities across two dimensions, see \cite{madeira}, \cite{pontes} for instance.

The problem also finds application in community identification, which aims to uncover clusters of nodes within a network that are more tightly connected to each other than to the rest of the network, revealing structures such as social groups or functional modules. This task is particularly complex in dynamic networks where relationships evolve over time. For instance, in \cite{tanti}, the authors address this challenge by using bicliques to track and identify changing communities, thereby providing insights into the evolving nature of social interactions.

\section{State of the art}

An algorithm is considered to have polynomial time delay if the time elapsed between the output of two bicliques is at most polynomial in the size of the input. The complexity of these algorithms is typically described in terms of the time delay. We refer to algorithms that exhibit complexity polynomial in the parameters of the graph and the output size, but are not classified as polynomial time delay, as polynomial \textit{output sensitive} algorithms. In the scope of this paper we will simply call such algorithms output sensitive.

There are a few contributions for the enumeration of maximal  non-induced bicliques. Eppstein~\cite{7} proposes a fixed parameter tractable algorithm running in time $\mathcal{O}(a^32^{2a}n)$ where $a$ is the arboricity of the input graph (the arboricity is a graph parameter within a constant factor of the degeneracy). Eppstein also proves that the maximum number of maximal non-induced bicliques is lower bounded by $\Omega(2^{a}n/a)$ and upper bounded by $\mathcal{O}(2^{2a}n)$. Alexe \textit{et al.}~\cite{8} propose various solutions for the problem. They propose an output sensitive algorithm with complexity $\mathcal{O}(n^2\alpha^2)$ with $n$ the order of the graph and $\alpha$ the number of maximal non induced bicliques. They also give a polynomial time delay algorithm with complexity $\mathcal{O}(n^3)$.  Additionnaly, Dias \textit{et al.}~\cite{14} propose an algorithm with complexity $\mathcal{O}(n^3)$.
 
Results also exist for the enumeration of maximal bicliques  when the input graph is subject to certain restrictions. For instance when it is bipartite, Makino and Uno~\cite{26} propose a polynomial time delay algorithm with delay $\mathcal{O}(\Delta^2)$ (with $\Delta$ the maximum degree of the graph).  Damaschke~\cite{16} also provides polynomial time delay algorithms for bipartite graphs with special degree distributions. Note that when the input graph is bipartite, all bicliques are necessarily induced.

Concerning the problem of detecting and counting maximum bicliques, there have been several contributions. For example, Gaspers \textit{et al.}~\cite{gaspers} present an $\mathcal{O}(1.2109^n)$ algorithm for finding a maximum biclique, as well as an $\mathcal{O}(1.2377^n)$ algorithm for counting all maximum bicliques.

\section{Our contributions}

In Section~\ref{fst}, we demonstrate that by leveraging existing combinatorial arguments and incorporating a few straightforward observations, we can develop an $\mathcal{O}(\Delta^2)$ output-sensitive algorithm for enumerating all maximal non-induced bicliques in general graphs. To the best of our knowledge, the best-known algorithm to date has complexity that is polynomial in the order of the graph and the output size. In Section~\ref{scnd}, we focus on detecting and counting maximum bicliques. We show that a maximum biclique in a graph $G$ can be found in time $\mathcal{O}(nc^{\Delta} \Delta^{\mathcal{O}(1)})$, given an algorithm to find a maximum biclique in time $\mathcal{O}(c^n n^{\mathcal{O}(1)})$. In practice, by applying a result from Gaspers \textit{et al.}~\cite{gaspers}, which shows that finding a maximum biclique is equivalent to finding a maximum independent set, and using, for example, the $\mathcal{O}(1.2109^n)$ algorithm of Robson~\cite{robson}, we obtain an $\mathcal{O}(n \cdot 1.2109^{\Delta})$ algorithm for this problem. This is stated in Theorem~\ref{algcom}. In a similar fashion, and using results from~\cite{gaspers,robson,sat}, we demonstrate that all maximum bicliques can be counted in time $\mathcal{O}(n\cdot 1.2377^{\Delta})$, and extend these results to bicliques of size $k$. These results are stated in Theorem~\ref{countt} and Corollary~\ref{countcor}.

\section{Results}
\subsection{Notations}

 We consider graphs of the form $G=(V,E)$ which are simple, undirected, with $n$ vertices and $m$ edges. If $X\subseteq V$, the subgraph of $G$ induced by $X$ is denoted by $G[X]$. If $X\subset V$, then $X$ is a proper subgraph of $G$.  When not clear from the context, the vertex set of $G$ will be denoted by $V(G)$. The set $N(x)$ is called the \textit{open neighborhood} of the vertex $x$ and consists of the vertices adjacent to $x$ in $G$. Given an ordering $\sigma= v_{1},...,v_{n}$ of the vertices of $G$, set $V_{i}$ consists of the vertices following $v_{i}$ including itself in this ordering, that is, the set $\{v_{i},v_{i+1},...,v_{n}\}$. In the ordering $\sigma$, the \textit{rank} of $v_i$, denoted by $\sigma(v_i)$, is its position in the ordering ($i$ in that case).  By $[n]$, we denote the set of integers $\{1,2,...,n\}$.  The distance between two vertices $u$, $v$ is the length of the shortest path from
$u$ to $v$. Let $N_{i}^k(v)=V_i \cap N^{k}(v)$ where $N^k(v)$ is the set of vertices at distance $k$ from $v$ (we consider in the paper that $N^{1}_i(v_i)$ and $N_{i}(v_i)$ are equal).

\subsection{Enumeration of non induced maximal bicliques}
\label{fst}
In this section we are interested in the problem of enumerating non-induced maximal bicliques. As discussed in the state of the art section, there are a few results for this question. It is important to notice that the problem can be related to the problem of enumerating induced bicliques in bipartite graphs. To see that, we introduce a construction that has been referenced in several papers in various forms. This construction links the problem of enumerating maximal non-induced bicliques in general graphs to the task of enumerating induced bicliques in bipartite graphs. The one we present has been introduced, to the best of our knowledge, in Dias \textit{et al.}~\cite{14}. Similar constructions have been proposed, for instance in for general graphs and in for  for oriented graphs. The construction of  \textit{et al.}~\cite{14} that we consider in this section is as follows.

\begin{definition}
\label{fdef}
Let $G=(V,E)$ be a graph. We construct a new bipartite graph $G_1=(V_1,E_1)$. Let $V'$ be a duplicate of the vertices of $V$. If $x\in V$ then let $x'\in V'$ be the corresponding vertex. The vertex set of $G_1$ is set $V\cup V'$. Its edge set is as follows. If $(x,y)\in E$ then $(x,y')\in E_1$ and $(x',y)\in E_1$. 
\end{definition}

If $G$ is the general graph and $G_1$ the graph transformed according to Definition~\ref{fdef}, the authors prove that there is a two-to-one correspondence between the induced maximal bicliques of $G_1$ and the non-induced maximal bicliques of $G$. Each non-induced biclique $B = X \cup Y$ of $G$ appears in $G_1$ as $B' = (V \cap X) \cup (V' \cap Y )$ and as $B'' = (V' \cap X) \cup (V \cap Y )$.

In the same paper~\cite{14}, the authors develop a polynomial delay algorithm for enumerating maximal induced bicliques in bipartite graphs, achieving a complexity of $\alpha \mathcal{O}(n^3)$, along with polynomial space complexity. They further assert that the same time and space complexities apply to general graphs, utilizing the construction mentioned in the previous paragraph.

However, no formal proof is provided for this claim. We identify several issues with this approach. First, it is unclear why the delay remains polynomial. It is possible that between the output of two bicliques in the original graph, numerous duplicates are generated, which would invalidate the polynomial time delay. Additionally, the authors do not explain how these duplicates are detected and therefore how the polynomial space complexity is maintained. If the solutions need to be stored for duplicate detection, the space complexity could become exponential in the worst case. Moreover, if their assertion that an algorithm for maximal biclique enumeration in bipartite graphs fully translates to the problem of maximal non-induced biclique enumeration in general graphs is valid, then applying the polynomial time delay algorithm by Makino and Uno~\cite{26} would result in an algorithm with a delay complexity of $\mathcal{O}(\Delta^2)$, which is an improvement over their $\mathcal{O}(n^3)$ complexity. This can been seen easily by observing that the transformed graph $G_1$ has the same order and maximum degree than $G$, asymptotically. Therefore, we show in the remainder of this section that it is possible to have an algorithm with running time $\mathcal{O}(\Delta^2)$ (per clique) and polynomial space. However no polynomial delay is proved. We will use very simple arguments and constructions.

Let $G$ be a general graph and $G_1$ its transformation, as per Definition~\ref{fdef}. First, using results from Dias \textit{et al.}~\cite{14}, we know that there is a two-to-one correspondence between the non induced maximal bicliques of $G$ and the (induced) maximal bicliques of $G_1$. Moreover $G_1$ is of order $\mathcal{O}(n)$ and has maximum degree $\mathcal{O}(\Delta)$. Therefore applying the results of Makino and Uno~\cite{26} to the bipartite graph $G_1$ gives an algorithm to find all non induced maximal bicliques of a graph $G$ in total time $\alpha \mathcal{O}(\Delta^2)$. However it is not clear that the delay and space remain polynomial, as discussed earlier. Therefore we show now that with an additional small procedure it is possible to get a polynomial space complexity.

 Given a graph $G=(V,E)$, let $G_1$ be defined as per Definition~\ref{fdef}. Assume $V(G_1)= V\cup V'$. We consider the following function. Let $f:(V\cup V')\xrightarrow{} V$ such that $f(x)=x$ if $x'\in V$ and $f(x')=x$ if $x\in V'$. Assume that when the algorithm of Makino and Uno~\cite{26} is running on graph $G_1$ it considers as some point some maximal biclique $K$. Assume that we are given a global ordering $\sigma$ on the vertices of $G$. Consider the set of vertices $S=f(K)$. Let $x$ be the vertex of $S$ with smallest rank in $\sigma$. If $f^{-1}(x) \in V$ then output $K$. Note that $f^{-1}(x)$ either belongs to $V$ or $V'$. This procedure will correctly output only one of the two copies of each non-induced biclique of $G$.

In summary, given an arbitrary graph $G$, applying the algorithm of Makino and Uno~\cite{26} to the graph $G_1$ described in Definition~\ref{fdef}, along with the simple procedure outlined in the previous paragraph, yields an algorithm that finds all non-induced maximal bicliques of $G$, without duplication and using polynomial space. However, it is unclear whether the delay remains polynomial. This algorithm runs in total time $\alpha \mathcal{O}(\Delta^2)$, with $\alpha$ the number of maximal non-induced bicliques of the input graph.

\subsection{Counting and maximum bicliques}

\label{scnd}
In this section, we focus on the detection of maximum bicliques. A biclique \(B\) in a graph is considered maximum if it is a largest (in terms of vertices) maximal biclique. Specifically, we are interested in detecting one such biclique, as outlined in Algorithm~\ref{algo}, and in counting all maximum bicliques, as discussed in Theorem~\ref{countt}.
First, given a graph \(G\), we define a family of induced subgraphs in the following definition. Then, in Corollary~\ref{main}, we demonstrate how maximal bicliques are related to these subgraphs.

\begin{definition}
\label{graph}
Let $G=(V,E)$ be a graph and let  $v_1, \dots, v_n $ be an ordering of its vertices. Graph $G_i,i\in [n]$, is the subgraph of $G$ induced on the vertex set $\{\{v_i\}\cup N_i(v_i)\cup N_{i}^2(v_i)\}$.
\end{definition}

\begin{definition}
    Let $G=(V,E)$ be a graph and let  $v_1, \dots, v_n $ be an ordering of its vertices. Let $G_i,i\in [n]$ be the graphs of Definition~\ref{graph}. Graphs $G_{i,k}$ with $ i\in [n], k\in[d(v_i)]$ are the subgraphs induced on the vertex set $\{\{v_i\}\cup N_i(v_i)\cup (N_i(x)\cap N_{i}^2(v_i))\}$ with $x\in N(v_i)$.
\end{definition}

\begin{lemma}
\label{vert}
 Let   $G_i,i\in [n]$ be an induced subgraph of $G$. Let $B$ be a maximum biclique of $G_i$. Then vertex $v_i$ is included in $V[B]$.
\end{lemma}

\begin{proof}
The proof is by definition since we have that $N_i(v_i)\subseteq N(v_i)$ and $N_{i}^2(v_i)\cap N(v_i)=\emptyset$. 
\end{proof}

\begin{lemma}
\label{atl1}
Let $G$ be a graph and $\sigma $ an ordering of its vertices. Let $B$ be a maximum biclique  of $G$. Biclique $B$ is an induced subgraph of graph $G_v$ where $v$ is the vertex of $B$ with lowest ranking in $\sigma$.
\end{lemma}

\begin{proof}

Assume that $B=X\cup Y$. Assume without loss of generality that $v\in X$. We have that $V(Y)\subseteq N_{\sigma(v)}(v)$ and $V(X)\backslash\{v\} \subseteq  N^{2}_{\sigma(v)}(v) $. Therefore, $V(B)\subseteq V(G_{\sigma(v)})$.
\end{proof}

\begin{lemma}
\label{atm1}
Let $G$ be a graph and $v_1,\dots,v_n$ an ordering of its vertices. Let $B$ be a maximum biclique of $G$. Let $v$ be the vertex of $V(B)$ with smallest ranking in $\sigma$. Then $B$ is not an induced graph of any $G_j$, with $j\neq \sigma(v)$.
\end{lemma}

\begin{proof}
    
 Assume by contradiction that there exists  $j\in [n]$ with $j\neq \sigma(v)$ such that $B$ is an induced subgraph of $G_j$. Assume first that $\sigma(v)<j$. By construction, vertex $v$ belongs to $V(B)$. Since we assumed $\sigma(v)<j$, then $v \notin V(G_j)$. This implies that $V(B)$ cannot be included in $V(G_j)$, which yields the contradiction. Assume now that $j<\sigma(v)$. By definition,  $B$ is a maximal biclique in $G_v$. Since $v_j\notin V(G_v)$, then $v_j \notin V(B)$. Therefore,  $B$ cannot be maximal in $G_j$, since  by Lemma~\ref{vert}, all maximal bicliques include vertex $v_j$.

\end{proof} 

\begin{corollary}
\label{main}
Let $G$ be a graph and $v_1,\dots,v_n$ an ordering of its vertices. Let $B$ be a maximum  biclique of $G$. Biclique $B$ is an induced graph of exactly one graph $G_i, i\in [n]$.
\end{corollary}

\begin{algorithm}

 \KwData{A graph $G$.}
 \KwResult{A maximum biclique of $G$.}
 \BlankLine
 
Consider any ordering $\sigma$ of the vertices of $G$.\label{orde} \\ Construct the graphs $G_{i},i\in [n]$.\label{famil}\\ Initialize $S$ an empty set.\\

\For{$i=1$ \KwTo $n$}{
\label{firstF}
Compute the graphs $G_{i,k}$ for $k\in [d(v_i)]$.

\For{each graph $G_{i,k}$}

{\label{comp}
Compute a maximum biclique $B$ of $G_{i,k}$.
 \If{\label{max}$B$ is larger that the biclique in $S$}{$S\leftarrow{\emptyset}, S\leftarrow B$.}
 
 } 
}
 \Return{S.}
 
\caption{}
\label{algo}
\end{algorithm}

Algorithm~\ref{algo} is a straightforward procedure that utilizes the properties established in the previous lemmas. It searches for maximum bicliques in the subgraphs defined in Definition~\ref{graph}, keeping track of the largest maximum biclique encountered during its execution. The algorithm is presented in Figure~\ref{algo}. We then prove its correctness in Theorem~\ref{corre} and analyze its time complexity in Theorem~\ref{algcom}.

\begin{theorem}
\label{corre}
    Algorithm~\ref{algo} correctly outputs a maximum biclique of the input graph $G$.
\end{theorem}

\begin{proof}
    We first show that $S$ cannot be empty when it is returned by Algorithm~\ref{algo}.
Let $B$ be a maximum biclique in $G$. By Corollary~\ref{main}, there exists $i$ such that biclique $B$ is an induced subgraph of graph $G_i$. There exists a graph $G_{i,k}$, with $k\in [d(v_i)]$, such that $B$ is an induced subgraph of $G_{i,k}$. To see that, it is enough to choose a vertex $v\in N_i(v_i)$ such that $v$ belongs to the vertex set of $B$. Then, we have that $B$ is and induced subgraph of a graph $G_{i,k}$, with $k\in [d(v_i)]$. Therefore, $S$ cannot be empty when it is output.
Now we show that $S$ contains necessarily a maximum biclique of the graph when it is output. Let $B$ be the biclique in $S$. Assume by contradiction that $B$ is not maximum in $G$. Then, there exists a larger biclique $B'$ which is maximum in $G$. There exists $i\in [n]$ and $k\in[d(v_i)]$ such that $B'$ is an induced subgraph of $G_{i,k}$. Therefore, $B'$ will be considered at some point in Line~\ref{max} of Algorithm~\ref{algo}. If $B$ is already in $S$ then it will be replaced by $B'$ otherwise $B'$ will be put in $S$ and $B'$ cannot replace it, by maximality.

\end{proof}

\begin{theorem}
\label{algcom}
Given an algorithm to find a maximum biclique of a graph of order $n$ in time $\mathcal{O}(c^n n^{\mathcal{O}(1)})$, Algorithm~\ref{algo} runs in time $\mathcal{O}(nc^{\Delta} \Delta^{\mathcal{O}(1)})$.
\end{theorem}

\begin{proof}

Algorithm~\ref{algo} finds maximum bicliques in the graphs $G_{i,k}$ defined in Theorem~\ref{corre}. There can be at most $\mathcal{O}(n\Delta)$ graphs \(G_{i,k}\), since \(i\) is bounded by \(\mathcal{O}(n)\) and \(k\) by \(\mathcal{O}(\Delta)\). Each of these graphs has a size of \(\mathcal{O}(\Delta + \Delta) = \mathcal{O}(\Delta)\). We then search for a maximum biclique in each of these graphs. Given an algorithm that runs in time \(\mathcal{O}(c^n n^{\mathcal{O}(1)})\) for a graph of order \(n\), we apply it to each of the graphs \(G_{i,k}\) of size \(\mathcal{O}(\Delta)\). Consequently, Algorithm~\ref{algo} operates in time \(\mathcal{O}(n\Delta c^{\Delta} \Delta^{\mathcal{O}(1)}) = \mathcal{O}(nc^{\Delta} \Delta^{\mathcal{O}(1)})\) overall.

\end{proof}

In the next theorem, we describe a simple algorithm for counting all maximum bicliques in a graph. As a corollary, this algorithm can also be used to count all bicliques of a given size $k$.

\begin{theorem}
\label{countt}
Given an algorithm to count all maximum bicliques  of a graph in time $\mathcal{O}(c^n n^{\mathcal{O}(1)})$, there exists an algorithm to count all maximum bicliques  of a graph in time $\mathcal{O}(nc^{\Delta} \Delta^{\mathcal{O}(1)})$.
\end{theorem}

\begin{proof}

We begin by outlining a straightforward algorithm that counts all maximum bicliques in time \(\mathcal{O}(n c^{\Delta^2} \Delta^{\mathcal{O}(1)})\), under the assumption that we have an existing algorithm for counting all maximum bicliques in time \(\mathcal{O}(c^n n^{\mathcal{O}(1)})\). Next, we demonstrate how to achieve the claimed complexity. 
To start, by applying Corollary~\ref{main}, counting the maximum bicliques in each graph \(G_i\) for \(i \in [n]\) will yield the total number of maximum bicliques in the overall graph. Each of these graphs has a size of at most \(\mathcal{O}(\Delta^2)\), and there are at most \(n\) such graphs. Consequently, the counting process across these graphs can be completed within the stated total time.

To achieve the complexity outlined in Theorem~\ref{countt}, we follow a process similar to that in Algorithm~\ref{algo}. First, we compute all $\mathcal{O}(n\Delta)$ graphs $G_{i,k}$ where~$i \in [n]$ and $k \in [d(v_i)]$. Next, we count the number of maximum bicliques in each of these graphs. Given an algorithm that counts maximum bicliques in a graph of order $n$ in time $\mathcal{O}(c^n n^{\mathcal{O}(1)})$, the overall complexity of this process is $\mathcal{O}(n c^{\Delta} \Delta^{\mathcal{O}(1)})$. Note that this algorithm does not give the exact number of maximum bicliques. To see this, consider some maximum biclique $B$. Let $V(B)=X\cup Y$ with $X$ and $Y$ being the two independent sets of vertices forming biclique $B$.

\end{proof}

\begin{corollary}
\label{countcor}
    Given an algorithm to count all  bicliques size $k$ of a graph in time $\mathcal{O}(c^n n^{\mathcal{O}(1)})$, there exists an algorithm to count all bicliques of size $k$ of a graph in time $\mathcal{O}(nc^{\Delta} \Delta^{\mathcal{O}(1)})$.
\end{corollary}

\begin{proof}

The proof is similar to that of Theorem~\ref{countt}, but instead of counting all maximum bicliques in the graphs \(G_{i,k}\) where \(i \in [n]\) and \(k \in [d(v_i)]\), we count all bicliques of size \(k-1\).

\end{proof}
\bibliographystyle{elsarticle-num}
\bibliography{biblio}

\end{document}